\documentclass[11pt]{article}             % Use with LaTeX2e
\usepackage{osid}                         %

\usepackage{amsmath,amssymb,amscd,amsthm,mathrsfs}
\usepackage{bbm}
\usepackage{graphicx}
\usepackage{enumerate} 
\usepackage{color,xcolor}
\usepackage{hyperref}
\usepackage{url}

\theoremstyle{plain}% Theorem-like structures provided by amsthm.sty
\newtheorem{lemma}{Lemma}
\newtheorem{thm}{Theorem}

\newtheorem{corollary}{Corollary}
\newtheorem{proposition}{Proposition}
\newtheorem{workflow}{Workflow}

\theoremstyle{definition}

\newtheorem{remark}{Remark}

\newcommand{\tr}{{\rm tr}}

\title{$k$-Positive Maps: New Characterizations and a Generation Method}
\author{Frederik vom Ende%\thanks{} 
	\\[1mm]{\footnotesize\it Dahlem Center for Complex Quantum
Systems, Freie Universität Berlin, Arnimallee 14, 14195 Berlin, Germany
 \& {frederik.vomende@gmail.com}}\\[2ex]
Sumeet Khatri%\thanks{} 
	\\[1mm]{\footnotesize\it Department of Computer Science, Virginia Tech, Blacksburg, VA 24061, USA and\\
    Virginia Tech Center for Quantum Information Science and Engineering, Blacksburg, VA 24061, USA
    \& {skhatri@vt.edu}}\\[2ex]	
Sergey Denisov%\thanks{} 
	\\[1mm]{\footnotesize\it Department of Computer Science, OsloMet -- Oslo Metropolitan University, P.O. Box 4 St. Olavs plass, N-0130 Oslo, Norway \& {sergiy.denysov@oslomet.no}}
}

\AtBeginDocument{
  \label{CorrectFirstPageLabel}
  
}

\begin{document}

\maketitle
\begin{abstract}
We study $k$-positive linear maps on matrix algebras and address two problems,  (i)  characterizations of $k$-positivity and (ii) generation of non-decomposable $k$-positive maps. On the characterization side, we derive optimization-based conditions equivalent to $k$-positivity that (a) reduce to a simple check when $k=d$, (b) reveal a direct link to the spectral norm of certain order-3 tensors (aligning with known NP-hardness barriers for $k<d$), and (c) recast $k$-positivity as a novel optimization problem over separable states, thereby connecting it explicitly to separability testing. On the generation side, we introduce a Lie-semigroup-based method that, starting from a single $k$-positive map, produces one-parameter families that remain $k$-positive and non-decomposable for small enough times. We illustrate this by generating such families for $d=3$ and $d=4$. We also formulate a semi-definite program (SDP) to test an equivalent form of the positive partial transpose (PPT) square conjecture (and do not find any violation of the latter). Our results provide practical computational tools for certifying $k$-positivity and a systematic way to sample $k$-positive non-decomposable maps.
\end{abstract}

\section{Introduction}

Positive linear maps lie at the heart of quantum information theory. They underpin entanglement witnesses and the separability problem \cite{HHHH07,GT09,CS14}, the divisibility of quantum dynamics \cite{Plenio_NMarkov_Review2014,Breuer16,LHW18}, as well as entanglement-breaking channels \cite{HSR03}, \cite[Ch.~6.4]{Holevo12_2nd_ed} (e.g., through certain formulations of the PPT-square conjecture, cf.~below).

However, positivity of a linear map is notoriously more difficult to characterize than the strictly stronger notion of complete positivity \cite{Woronowicz76,Stormer13}, with the latter being famously characterized by the positive semi-definiteness of the associated Choi matrix \cite{Choi75}. Recalling that a bipartite quantum state $\rho$ is separable if and only if $({\rm id}\otimes\Phi)(\rho)\geq 0$ for all positive maps $\Phi$ \cite{HHH96,Stormer86}, this characterization difficulty is also reflected in the fact that the separability decision problem is NP-hard \cite{Gurvits03,Gharibian10}. This may explain why, to our knowledge, there are currently only a few reliable ways to numerically test whether a state is separable, whether an operator is an entanglement witness, or whether a map is positive -- let alone numerical methods to generate witnesses or positive but not completely positive maps (a notable exception is the recent work on numerical witness construction \cite{LHPD25}).

In addition to the absence of any reliable method for generating non-decomposable entanglement witnesses 
(resp., non-decomposable $k$-positive maps), testing conjectures that rely on these notions---such as the PPT-square conjecture or the Sanpera--Bru\ss--Lewenstein conjecture---is challenging. To clarify what this entails and to motivate our work, we briefly review both.

The PPT-square conjecture (attributed to Christandl and first formulated in a 2012 workshop report \cite{RuskaiEtAl2012}) states that the composition of two maps that are completely positive and completely co-positive (i.e., whose composition with the transpose map is completely positive) is always entanglement breaking. %This conjecture is relevant to certain questions in quantum communication \cite{BCHW15,CF17}.
Although the conjecture has been proved in many special cases, including the $d=3$ case~\cite{CMW19,CYT19}, in asymptotic \cite{KMP18,RJP18} and random regimes \cite{CYZ18,NP25}, as well as for all Choi-type maps \cite{SN22} and Gaussian channels \cite{CMW19}, the question of its validity in general case remains open.

This conjecture fits our setting via an equivalent formulation: The PPT-square conjecture holds if and only if, for any completely positive and completely co-positive map $\Delta$ and any positive map $\Psi$, the composition $\Psi\circ\Delta$ is decomposable \cite[Conjecture~IV.3]{CMW19}.\footnote{A map $\Phi$ is called decomposable if there exist completely positive maps $\Phi_1,\Phi_2$ such that $\Phi=\Phi_1+\Phi_2\circ T$, where $T$ denotes the transpose map. Otherwise, $\Phi$ is called non-decomposable.} Notably, only \emph{non-decomposable} positive maps $\Psi$ are relevant here, because if $\Psi$ is decomposable then $\Psi\circ\Delta$ is trivially decomposable (as $\Delta$ is completely positive by assumption).
Thus, numerically testing this formulation of the PPT-square conjecture requires a method to generate non-decomposable positive maps. Moreover, since a potential---but so far unverified---counterexample to PPT-square has recently been proposed \cite[Ex.~7]{JY20}, this seems like a promising path forward.

A strengthened form of the Sanpera--Bru\ss--Lewenstein (SBL) conjecture asserts that, in dimension $d$, the Choi matrix of every map that is both completely positive and completely co-positive (i.e., PPT) has Schmidt number at most $d{-}1$~\cite{SBL01,CMW19}. By duality of cones, this is equivalent to the statement that every $(d{-}1)$-positive map on $d$-dimensional matrices is decomposable~\cite{Yang16}. This equivalence is known to hold in $d=3$. 
An explicit counterexample in dimension $d\ge 4$ (either a PPT Choi matrix with Schmidt number $d$, or, equivalently, a $(d{-}1)$-positive map that is non-decomposable) would falsify the SBL-type statement. Such an example could also lead to a counterexample to the PPT-square conjecture, since a PPT state of full Schmidt number is far from being entanglement breaking~\cite{MH25_priv}.

All of the above considerations lead to the problem of devising methods to sample non-decomposable $k$-positive maps and, more fundamentally, of identifying conditions that make testing $k$-positivity numerically as simple as possible -- despite its likely NP-hardness for $k<d$. This is our motivation.

The paper is structured as follows. In Section~\ref{sec_prelim} we recap known characterizations of $k$-positivity, including a recent characterization by Marciniak et al.~\cite{MOM25}. Building on this, Section~\ref{sec_mainres} develops new characterizations of $k$-positivity via several optimization problems. These yield new connections to the problem of computing the spectral norm of 3-tensors (Section~\ref{sec_3tensor}) and to the separability problem (Section~\ref{sec_newconnect_ent}). In Section~\ref{sec_newnondec} we introduce a Lie-semigroup-based method that, starting from a given $k$-positive map satisfying suitable conditions, generates further $k$-positive maps. We demonstrate the method by numerically producing one-parameter families of non-decomposable, positive maps in dimensions $d=3$ and $d=4$. We then use this new method to test the previously mentioned formulation of the PPT-square conjecture involving non-decomposable positive maps in Section~\ref{sec_ppt2}. In fact, given such a map, we formulate a test for counterexamples to the conjecture as a semi-definite program (Eq.~\eqref{eq:PPT2_SDP}). 
Finally, we conclude in Section~\ref{sec_outlook}.

\section{Recap: $k$-positivity and its characterizations}\label{sec_prelim}

Let $\mathcal L(\mathbb C^{m\times m},\mathbb C^{d\times d})$ denote the set of all linear maps from $\mathbb C^{m\times m}$ to $\mathbb C^{d\times d}$. Recall the following definition: $\Phi\in\mathcal L(\mathbb C^{m\times m},\mathbb C^{d\times d})$ is called $k$-positive if for all $A$ positive semi-definite (written: $A\geq 0$) it holds that $({\rm id}_k\otimes\Phi)(A)\geq 0$, as well. This property is known to be encoded in the Choi matrix $$\mathsf C(\Phi):= ({\rm id}\otimes\Phi)(|\Gamma\rangle\langle\Gamma|) =\sum_{j,k=1}^m|j\rangle\langle k|\otimes\Phi(|j\rangle\langle k|)$$ (where $|\Gamma\rangle:=\sum_j|j\rangle\otimes|j\rangle$), as well as in the representation matrix $\widehat\Phi$ (i.e., the unique matrix such that ${\rm vec}(\Phi(A))=\widehat\Phi{\rm vec}(A)$ for all $A$ where ${\rm vec}(A):=\sum_j |j\rangle\otimes A|j\rangle$ is usual column-vectorization \cite{HS81,HJ2}) in the following way:
\begin{proposition}\label{prop_CP_P_char}
Given any $\Phi\in\mathcal L(\mathbb C^{m\times m},\mathbb C^{d\times d})$, $k,m,d\in\mathbb N$ the following statements are equivalent.
\begin{itemize}
    \item[(i)] $\Phi$ is $k$-positive
    \item[(ii)] $\langle\psi|\mathsf C(\Phi)|\psi\rangle\geq 0$ for all $\psi\in\mathbb C^m\otimes\mathbb C^d$ of Schmidt rank\footnote{i.e., the number of non-vanishing coefficients in the Schmidt decomposition.} at most $k$.
    \item[(iii)] $\langle\overline{X}\otimes X,\widehat\Phi\rangle_{\rm HS}\geq 0$ for all $X\in\mathbb C^{m\times d}$ with rank at most $k$. Here $\langle A,B\rangle_{\rm HS}:=\tr(A^\dagger B)$ is the usual Hilbert-Schmidt inner product.
    \item[(iv)] $\tr(\Phi(X(\cdot)X^\dagger))\geq 0$ for all $X\in\mathbb C^{m\times d}$ with rank at most $k$, with ${\rm tr}:=\sum_j\langle G_j,(\cdot)(G_j)\rangle_{\rm HS}$ the usual trace of a linear map.
    \item[(v)] $\tr(\Psi^\dagger\Phi)\geq 0$ for all $\Psi\in\mathcal L(\mathbb C^{m\times m},\mathbb C^{d\times d})$ completely positive such that $\Psi$ admits a set of Kraus operators of rank at most $k$.
\end{itemize}
\end{proposition}
For the sake of completeness, a proof can be found in Appendix~A. More importantly, some remarks are in order:
\begin{remark}\label{rem1}
\begin{itemize}
\item[(i)] Condition~(ii) for $k=1$ is also known as ``block positivity'' of $\mathsf C(\Phi)$, i.e., $\langle\psi\otimes\phi|\mathsf C(\Phi)|\psi\otimes\phi\rangle\geq 0$ for all $\psi,\phi$ as first found by Jamio\l{}kowski \cite{Jamiolkowski72}. This condition can equivalently be rephrased using the so-called positivity polynomial of $\mathsf C(\Phi)$, which can be checked using Renegar's quantifier elimination technique \cite[Algorithm~5.7]{PSJ22}, \cite{PSJ21}. As is well-known from complexity results in real algebraic geometry, this algorithm often has at least exponential---and in many cases even doubly exponential---runtime in the worst case \cite[Ch.~14]{BPR06}.
    \item[(ii)] In the case of one qubit it is known that the Choi matrix of a positive map can have at most one negative eigenvalue \cite{STV98}. For general bipartite systems it has been shown that the partial transpose of states on $\mathbb C^m\otimes\mathbb C^d$ can have at most $(m-1)(d-1)$ negative eigenvalues \cite{Rana13,Johnston13}, which reduces to the previously mentioned qubit result when $m=d=2$. Beyond this, we are not aware of any further eigenvalue-sign constraints on the Choi matrix of $k$-positive maps for $1<k<d$.
    \item[(iii)] Condition~(v) in the previous proposition shows that the cone of $k$-positive maps is dual to the cone of completely positive maps which admit a set of rank-$k$ Kraus operators. For $k=d$ this reproduces the known result that the cone of completely positive maps is self-dual \cite[Ch.~11.2]{Bengtsson17}, while for $k=1$ this shows that the positive maps are dual to the entanglement-breaking channels \cite[Prop.~4.14]{KW20}.
    \item[(iv)] A simple \textit{sufficient} way to check positivity is to check for complete positivity after applying the inverse of a given positive, bijective map
    \cite{Mozrzymas15,ChruscinskiComment16}.
    After a slight reformulation and generalization this result reads as follows: Given $\Phi,\Psi\in\mathcal L(\mathbb C^{d\times d})$ such that $\Phi$ is positive and bijective, then complete positivity of $\Phi^{-1}\circ\Psi$ (which is easy to check) implies that $\Psi$ is positive (simply because $\Psi=\Phi\circ(\Phi^{-1}\circ\Psi)$ positive by assumption). 
    The example used in the cited papers is the reduction map $R(X):=\tr(X){\bf1}-X$ with inverse $R^{-1}(X)=\frac{1}{d-1}\tr(X){\bf1}-X$.
\end{itemize}
\end{remark}

A different perspective on characterizing $k$-positivity was proposed by Marciniak et al.~\cite[Thm.~3.2]{MOM25}. Instead of testing the Choi matrix $\mathsf C(\Phi)$ itself, their method considers a shifted---hence necessarily completely positive---variant of $\Phi$, in the spirit of St{\o}rmer’s construction~\cite{Stormer11}.

\begin{proposition}\label{prop_Marcin}
 Given any $d\in\mathbb N$, $k\in\{1,\ldots,d\}$, and $\Phi\in\mathcal L(\mathbb C^{d\times d})$ Hermitian-preserving\footnote{i.e., if $A$ is Hermitian, then so is $\Phi(A)$. Equivalently, $\mathsf C(\Phi)$ is Hermitian.} the following statements are equivalent.
 \begin{itemize}
 \item[(i)] $\Phi$ is $k$-positive
 \item[(ii)]
 \begin{equation}\label{eq:Marcin1}
 \max( \sigma( \mathsf C(\Phi) ) )\geq \max_{x\in\mathbb C^r,\|x\|=1}\Big\| \Big(\sum_ix_iK_i\Big)\Big(\sum_jx_jK_j\Big)^\dagger \Big\|_{(k)}
 \end{equation}
 where $\sigma( \mathsf C(\Phi) )$ is the spectrum of the Choi matrix of $\Phi$, $\|\cdot\|_{(k)}$ is the $k$-th Ky Fan norm (i.e., the sum of the $k$ largest singular values),
 $r$ is the Kraus rank of $\max( \sigma( \mathsf C(\Phi) ) )\tr(\cdot){\bf1}-\Phi$, and $\{K_j\}_{j=1}^r$ is any set of Kraus operators of $\max( \sigma( \mathsf C(\Phi) ) )\tr(\cdot){\bf1}-\Phi$.
 \end{itemize}
\end{proposition}
Here, some remarks are in order.
\begin{remark}\label{rem_marc}
    \begin{itemize}
    \item[(i)] The maximization~\eqref{eq:Marcin1} is done over the $r$-dimensional unit sphere, where, generically, $r=d^2-1$. Ideally, for the purpose of numerics one may want an optimization problem where the feasible set has dimension $O(d)$, or a problem where the feasible set has nicer topological properties (e.g., convex, compact).
    \item[(ii)] The function to be maximized in~\eqref{eq:Marcin1} is the norm (i.e., non-smooth) of a function which is quadratic in $x$. This raises the question whether there exist other such characterizations where the function inside the norm is linear, or where the entire objective function is linear (or at least smooth).
    \item[(iii)]
    Unlike Prop.~\ref{prop_CP_P_char}, condition~\eqref{eq:Marcin1} does not---or at least not obviously---reduce to a ``simple'' problem if $k=d$. This begs the question whether there are characterizations of this form where the corresponding optimization problem trivializes when checking for complete positivity.
    \item[(iv)]
    Connected to the previous point, \eqref{eq:Marcin1} can be re-written slightly in the following way when defining $V:=\sum_iK_i\otimes|i\rangle$:
    \begin{align*}
\sum_ix_iK_i&=\sum_iK_i\otimes\overline{\langle x|i\rangle}=\sum_iK_i\otimes\langle \overline{x}|i\rangle\\
&=({\bf1}\otimes\langle \overline x|)\Big(\sum_iK_i\otimes|i\rangle\Big)=({\bf1}\otimes\langle \overline x|)V\,,
    \end{align*}
i.e., the maximum in~\eqref{eq:Marcin1} equals
$\max_{x\in\mathbb C^r,\|x\|=1}\| ({\bf1}\otimes\langle  x|)VV^\dagger ({\bf1}\otimes|  x\rangle) \|_{(k)}$.
For usual positivity ($k=1$) this reduces further to
\begin{equation*}
    \max_{x\in\mathbb C^r,\|x\|=1}\| ({\bf1}\otimes\langle  x|)VV^\dagger ({\bf1}\otimes|  x\rangle) \|_\infty=\max_{x\in\mathbb C^r,\|x\|=1}\| V^\dagger ({\bf1}\otimes|  x\rangle) \|_\infty^2
\end{equation*}
using the $C^*$-identity $\|AA^\dagger\|_\infty=\|A\|_\infty^2$.
Hence the Hermitian-preserving map $\Phi$ is positive if and only if
\begin{equation}\label{eq:Marc_1_3}
    \max_{x\in\mathbb C^r,\|x\|=1}\| V^\dagger ({\bf1}\otimes|  x\rangle) \|_\infty\leq\sqrt{\max( \sigma( \mathsf C(\Phi) ) )}\,.
\end{equation}
While this turns the function inside the norm into a linear function, this trick does not work beyond usual positivity, i.e., it breaks down as soon as $k\geq 2$ because of $C^*$-identity does not generalize to higher Ky Fan-norms.
    \end{itemize}
\end{remark}
With this in mind, our next goal is to explore other conditions which are equivalent to~\eqref{eq:Marcin1} which perhaps amend or improve upon some of these points.

\section{Novel characterizations of $k$-positivity}\label{sec_mainres}
\subsection{First condition: Collapse for $k=d$}
In order to understand what happens to the maximum in~\eqref{eq:Marcin1} for the special case $k=d$ we need the following characterization of the Ky Fan-norm, a proof of which can be found in Appendix~B.
\begin{lemma}\label{lemma_char_KyFan}
 Given $d\in\mathbb N$, $k\in\{1,\ldots,d\}$, and $B\in\mathbb C^{d\times d}$ positive semi-definite it holds that
 $
 \|B\|_{(k)}=\max_{\substack{M\in\mathbb C^{d\times d}\\0\leq M\leq{\bf1},\tr(M)\leq k}}
 \tr(BM)
 $.
\end{lemma}
This allows us to find the following re-formulation of Prop.~\ref{prop_Marcin} which will resolve the question raised in Rem.~\ref{rem_marc}~(iii).
\begin{thm}\label{thmmain}
 Given any $d\in\mathbb N$, $k\in\{1,\ldots,d\}$, and $\Phi\in\mathcal L(\mathbb C^{d\times d})$ Hermitian-preserving the following statements are equivalent.
 \begin{itemize}
 \item[(i)] $\Phi$ is $k$-positive
 \item[(ii)]
 \begin{equation}\label{eq:equivineq}
 \max( \sigma( \mathsf C(\Phi) ) )\geq \max_{\substack{M\in\mathbb C^{d\times d}\\
 0\leq M\leq{\bf1},\tr(M)\leq k
 }}\big\| 
 \tr_{1,M}(VV^\dagger)
 \big\|_{\infty}
 \end{equation}
 where
 $V:\mathbb C^d\to\mathbb C^d\otimes\mathbb C^r$ is any Stinespring operator\footnote{
 The convention we choose here and henceforth is that $V$ is a Stinespring operator of $\Phi$ if $\Phi=\tr_2(V(\cdot)V^\dagger)$.
 }
 of $\max( \sigma( \mathsf C(\Phi) ) )\tr(\cdot){\bf1}-\Phi$ (with $r\leq d^2$ the rank of $\max( \sigma( \mathsf C(\Phi) ) ){\bf1}-\mathsf C(\Phi)$), and $\tr_{1,X}(Y)$ is the partial trace of $Y$ with respect to $X$, i.e., $\tr_{1,X}(Y)$ is the unique operator such that $\tr(Z\tr_{1,X}(Y))=\tr((X\otimes Z)Y)$ for all $Z$\,\footnote{
 The relation to the standard partial trace reads $\tr_{1,X}(Y)=\tr_1(( X\otimes{\bf1})Y)$
 as is straightforward from the definitions.
 }.
 \end{itemize}
\end{thm}
\begin{proof}
The first step is to re-write the maximum in~\eqref{eq:equivineq}: Given any $A,B\geq 0$ of compatible size, $\tr_{1,A}(B)\geq 0$ (as follows readily from the definition), so we compute
\begin{align*}
 \max_{0\leq M\leq{\bf1},\tr(M)\leq k
 }\big\| 
 \tr_{1,M}(VV^\dagger)
 \big\|_{\infty}&=\max_{0\leq M\leq{\bf1},\tr(M)\leq k
 }\max_{\|x\|=1}\langle x| 
 \tr_{1,M}(VV^\dagger)|x\rangle\\
 &=\max_{0\leq M\leq{\bf1},\tr(M)\leq k
 }\max_{\|x\|=1}\tr\big(|x\rangle\langle x| 
 \tr_{1,M}(VV^\dagger)\big)\\
 &=\max_{0\leq M\leq{\bf1},\tr(M)\leq k
 }\max_{\|x\|=1} 
 \tr\big((M \otimes  |x\rangle\langle x|)VV^\dagger\big)\\
 &=\max_{\|x\|=1} \max_{0\leq M\leq{\bf1},\tr(M)\leq k
 }
 \tr\big((M \otimes  |x\rangle\langle x|)VV^\dagger\big)\\
 &=\max_{\|x\|=1} \max_{0\leq M\leq{\bf1},\tr(M)\leq k
 }
 \tr\big( M\tr_{2,|x\rangle\langle x|}(VV^\dagger)\big)\,.
\end{align*}
Here, $\tr_{2,|x\rangle\langle x|}(VV^\dagger)$ is defined analogously to $\tr_{1,(\cdot)}$. Moreover, by the same argument as before $\tr_{2,|x\rangle\langle x|}(VV^\dagger)\geq 0$.
This re-formulation is useful because we can now combine it with Lemma~\ref{lemma_char_KyFan} to obtain
\begin{align*}
 \max_{0\leq M\leq{\bf1},\tr(M)\leq k
 }\big\| 
 \tr_{1,M}(VV^\dagger)
 \big\|_{\infty}&=
 \max_{\|x\|=1} \big\|\tr_{2,|x\rangle\langle x|}(VV^\dagger)\big\|_{(k)}\,.
\end{align*}
With this we are ready to prove the claimed equivalence.

``(ii) $\Rightarrow$ (i)'': To show that $\Phi$ is $k$-positive we have to verify~\eqref{eq:Marcin1}. Given any set of Kraus operators $\{K_j\}_j$ of $\max( \sigma( \mathsf C(\Phi) ) )\tr(\cdot){\bf1}-\Phi$, $V:=\sum_j K_j\otimes |j\rangle $ is well known to be a Stinespring operator of $\max( \sigma( \mathsf C(\Phi) ) )\tr(\cdot){\bf1}-\Phi$. 
Thus
\begin{align*}
 \max( \sigma( \mathsf C(\Phi) ) )&\geq\max_{0\leq M\leq{\bf1},\tr(M)\leq k
 }\big\| 
 \tr_{1,M}(VV^\dagger)
 \big\|_{\infty}\\
 &=
 \max_{\|x\|=1} \big\|\tr_{2,|x\rangle\langle x|}(VV^\dagger)\big\|_{(k)}\\
 &=\max_{\|x\|=1} \Big\|\sum_{i,j}\tr_{2,|x\rangle\langle x|}( K_iK_j^\dagger\otimes|i\rangle\langle j|)\Big\|_{(k)}\\
 &=\max_{\|x\|=1} \Big\|\sum_{i,j}\langle x|i\rangle\langle j|x\rangle K_iK_j^\dagger)\Big\|_{(k)}\\
 &=\max_{\|x\|=1}\Big\| \Big(\sum_i\overline x_iK_i\Big)\Big(\sum_j\overline x_jK_j\Big)^\dagger \Big\|_{(k)}\\
 &=\max_{\|x\|=1}\Big\| \Big(\sum_ix_iK_i\Big)\Big(\sum_jx_jK_j\Big)^\dagger \Big\|_{(k)}
\end{align*}
where in the last step we substituted $x\to\overline x$.
Finally, ``(i) $\Rightarrow$ (ii)''
is shown analogously.
\end{proof}

What is nice about this criterion is that for $k\geq d$ (i.e., complete positivity) it illustrates how the optimization problem in question resolves in a most simple manner: If $k\geq d$, then every operator $M\in\mathbb C^{d\times d}$ with $0\leq M\leq{\bf1}$ automatically satisfies $\tr(M)\leq k$, so
\begin{align*}
    \max_{ 0\leq M\leq{\bf1},\tr(M)\leq k
 }\big\| 
 \tr_{1,M}(VV^\dagger)
 \big\|_{\infty}&\leq\max_{\|M\|_\infty\leq 1}\|\tr_{1,M}(VV^\dagger)\|_\infty\\
 &=\|\tr_{1,(\cdot)}(VV^\dagger)\|_{\infty\to\infty}\\
 &=\|\tr_{1}(VV^\dagger)\|_\infty\\
 &\leq \max_{ 0\leq M\leq{\bf1},\tr(M)\leq k
 }\big\| 
 \tr_{1,M}(VV^\dagger)
 \big\|_{\infty}\,.
\end{align*}
In the third step we used \cite[Coro.~2.3.8]{Bhatia07} because $M\mapsto \tr_{1,M}(VV^\dagger)$ is a positive map, so its norm is attained on the identity.
In other words, what all of this shows is that for $k\geq d$, the maximum in~\eqref{eq:Marcin1} is equal to the (easily computable) quantity $\|\tr_{1}(VV^\dagger)\|_\infty$ where $V=\sum_j K_j\otimes|j\rangle$.
Hence, $\Phi$ is completely positive if and only if the largest eigenvalue of $\tr_{1}(VV^\dagger)$ does not exceed $\max( \sigma( \mathsf C(\Phi) ) )$.

\subsection{Second condition: Connection to 3-tensors}\label{sec_3tensor}

Next, let us follow up on the observation from Rem.~\ref{rem_marc}~(iv) and attempt to generalize Eq.~\eqref{eq:Marc_1_3} such that it also characterizes $k$-positivity for $k>1$.
In fact, we will prove the following.
\begin{thm}\label{thm1a}
 Given any $d\in\mathbb N$, $k\in\{1,\ldots,d\}$, and $\Phi\in\mathcal L(\mathbb C^{d\times d})$ Hermitian-preserving the following statements are equivalent.
 \begin{itemize}
 \item[(i)] $\Phi$ is $k$-positive
 \item[(ii)]
\begin{equation}\label{eq:thm1a_1}
    \max_{x\in\mathbb C^d\otimes\mathbb C^k,\|x\|=1}\|V^\dagger(| x\rangle\otimes{\bf1})\|_\infty\leq \sqrt{ k\max( \sigma( \mathsf C(\Phi) ) )}
\end{equation}
 where $V:\mathbb C^d\otimes\mathbb C^k\to\mathbb C^d\otimes\mathbb C^k\otimes\mathbb C^r$ is any Stinespring operator
 of $k\max( \sigma( \mathsf C(\Phi) ) )\tr(\cdot){\bf1}-\Phi\otimes{\rm id}_k$ (with $r\leq d^2$ the rank of $\max( \sigma( \mathsf C(\Phi) ) ){\bf1}-\mathsf C(\Phi)$).
 \end{itemize}
\end{thm}
A proof can be found in Appendix~C. At this point let us note two things. First, while~\eqref{eq:thm1a_1} looks quite similar to~\eqref{eq:Marc_1_3} (and for $k=1$ it is easy to see that the two are equivalent, so this is in fact a generalization of the latter), for $k>1$ this equivalence is far from obvious because the $k$-th Ky-Fan norm of $AA^\dagger$ can, in general, not be expressed via the $k$-th Ky Fan norm of $A$.
Moreover, while this result is equivalent to the one by Marciniak et al., our proof differs fundamentally from theirs, as we did not employ the quantity $\|({\bf1}\otimes Q)C({\bf1}\otimes Q)\|$ (with $Q$ projections of rank $k$) at any point.

Now one of the reasons Eq.~\eqref{eq:thm1a_1} is conceptually interesting is that the operator $V^\dagger((\cdot)\otimes{\bf1})$ is the coordinate hypermatrix of the trilinear map
\begin{align*}
    V':\mathbb C^{nk}\times\mathbb C^r\times\mathbb C^{nk}&\to\mathbb C\\
    (x,y,z)&\mapsto z^TV^\dagger (x\otimes y)\,,
\end{align*}
making it basically a 3-tensor.
The spectral norm of this map is defined as $\|V'\|:=\sup_{\|x\|=\|y\|=\|z\|=1} |z^T V^\dagger (x\otimes y)|$, which is readily verified to be equal to the optimization problem in Eq.~\eqref{eq:thm1a_1}.
However, computing and even \textit{approximating} such a spectral norm is NP-hard \cite[Thm.~1.11 \& 1.12]{HL13}, which further confirms that deciding $k$-positivity for $k<d$ is most likely an NP-hard problem as well.
We note that $\|V'\|$ is also known as the Schmidt 1-norm of $V^\dagger$ (or of $VV^\dagger$, for symmetric dimensions) \cite{JK10} which has a semi-definite program (SDP) hierarchy associated to it \cite{JK11}.

\subsection{Third condition: Connection to the separability problem}\label{sec_newconnect_ent}

The section's final conceptual insight will be about finding the separability problem is the previous optimization-based characterizations of $k$-positivity. While this connection is not surprising given that separability is characterized via positive maps (as already explained in the introduction) \cite{HHH96,Stormer86}, it will still be interesting to make this connection explicit. In fact, we get the following corollary of the previous result:

\begin{corollary}\label{coro_main1}
 Given any $d\in\mathbb N$, $k\in\{1,\ldots,d\}$, and $\Phi\in\mathcal L(\mathbb C^{d\times d})$ Hermitian-preserving the following statements are equivalent.
 \begin{itemize}
 \item[(i)] $\Phi$ is $k$-positive
 \item[(ii)]
 \begin{equation*}
\max_{\rho\in\mathbb D(\mathbb C^d\otimes\mathbb C^k),\omega\in\mathbb D(\mathbb C^r)}\langle VV^\dagger,\rho\otimes\omega\rangle_{\rm HS}
 \leq
 k\max( \sigma( \mathsf C(\Phi) ) )
 \end{equation*}
 where $V:\mathbb C^d\otimes\mathbb C^k\to\mathbb C^d\otimes\mathbb C^k\otimes\mathbb C^r$ is any Stinespring operator
 of $k\max( \sigma( \mathsf C(\Phi) ) )\tr(\cdot){\bf1}-\Phi\otimes{\rm id}_k$ (with $r\leq d^2$ the rank of $\max( \sigma( \mathsf C(\Phi) ) ){\bf1}-\mathsf C(\Phi)$).
 Moreover, here and henceforth, $\mathbb{D}$ shall denote the set of density operators over a given Hilbert space.
 \item[(iii)] $\max_{\rho\in\mathbb D(\mathbb C^{nk}\otimes\mathbb C^r)\text{ separable}}\langle VV^\dagger,\rho\rangle_{\rm HS}
 \leq
 k\max( \sigma( \mathsf C(\Phi) ) )$ with $V$ as in (ii).
 \end{itemize}
\end{corollary}
\begin{proof}
    ``(i) $\Leftrightarrow$ (ii)'': Because $\tr\big(V^\dagger(\rho\otimes\omega)V\big)$ is a convex (because linear) function in both $\rho$ and $\omega$, which is optimized over the compact, convex set $\mathbb D(\mathbb C^d\otimes\mathbb C^k)\times \mathbb D(\mathbb C^r)$, the maximum is attained on an extreme point. In other words,
    $$
    \max_{\rho\in\mathbb D(\mathbb C^d\otimes\mathbb C^k),\omega\in\mathbb D(\mathbb C^r)}\tr\big(V^\dagger(\rho\otimes\omega)V\big)=\max_{\substack{x\in\mathbb C^d\otimes\mathbb C^k,y\in \mathbb C^r\\\|x\|=\|y\|=1}}\tr\big(V^\dagger(|x\rangle\langle x|\otimes|y\rangle\langle y|)V\big)\,.
    $$
    Our goal is to show that this is equal to $\max_{x\in\mathbb C^d\otimes\mathbb C^k,\|x\|=1}\|V^\dagger(|x\rangle\otimes{\bf1})\|_\infty^2$ because then we would be done by Theorem~\ref{thm1a}. Indeed,
    \begin{align*}
        \max_{x\in\mathbb C^d\otimes\mathbb C^k,\|x\|=1}&\max_{y\in \mathbb C^r,\|y\|=1}\tr\big(V^\dagger(|x\rangle\langle x|\otimes|y\rangle\langle y|)V\big)\\
        &=\max_{x\in\mathbb C^d\otimes\mathbb C^k,\|x\|=1}\max_{y\in \mathbb C^r,\|y\|=1}\tr\big(V^\dagger(|x\rangle\otimes{\bf1})|y\rangle\langle y|(\langle x|\otimes{\bf1})V\big)\\
        &=\max_{x\in\mathbb C^d\otimes\mathbb C^k,\|x\|=1}\max_{y\in \mathbb C^r,\|y\|=1}\langle y|(\langle x|\otimes{\bf1})VV^\dagger(|x\rangle\otimes{\bf1})|y\rangle\\
        &=\max_{x\in\mathbb C^d\otimes\mathbb C^k,\|x\|=1}\|(\langle x|\otimes{\bf1})VV^\dagger(|x\rangle\otimes{\bf1})\|_\infty\\
        &=\max_{x\in\mathbb C^d\otimes\mathbb C^k,\|x\|=1}\|V^\dagger(|x\rangle\otimes{\bf1})\|_\infty^2\,.
    \end{align*}
    In the third step we used that $(\langle x|\otimes{\bf1})VV^\dagger(|x\rangle\otimes{\bf1})\geq 0$ so the maximum over $y$ yields the largest eigenvalue---which, here, is equivalent to the largest singular value (operator norm)---and in the last step we used the $C^*$-property $\|AA^\dagger\|_\infty=\|A\|_\infty^2$. ``(ii) $\Leftrightarrow$ (iii)'': This now follows at once from convexifying the problem.
\end{proof}
In other words, the separability problem was implicitly hidden in the optimization's feasible set. While, intuitively, this may be ``how'' the NP-hardness of entanglement detection translates to hardness of deciding ($k$-)positivity, one has to be careful with such conclusions because for $k=d$ this argument somehow has to break down (as deciding complete positivity is certainly a lot simpler than deciding separability).

At point we want to stress that, unsurprisingly, Coro.~\ref{coro_main1} is not the only way to encode $k$-positivity in a bilinear optimization problem: recall that $\Phi$ is positive if and only if ${\rm tr}(\mathsf C(\Phi)(\rho\otimes\omega))\geq 0$ for all states $\rho,\omega$ (block positivity). Thus $k$-positivity of a Hermitian-preserving linear map $\Phi:\mathbb C^{d\times d}\to\mathbb C^{d\times d}$ is equivalent to
\begin{equation}\label{eq:kpos_bilin_alt}
    \min_{\rho,\omega\in\mathbb D(\mathbb C^d\otimes\mathbb C^k)}{\rm tr}(\mathsf C(\Phi\otimes{\rm id}_k)(\rho\otimes\omega))\geq 0\,.
\end{equation}
One reason these bilinear formulations are interesting is that there already exist tools to compute them numerically.
Aside from the well-established seesaw method or branch-and-bound algorithms \cite{HKT20}, another way to tackle this is via SDP hierarchies for bilinear optimization problems \cite{berta2022bilinearSDPhierarchy}.
In fact, applying the first level of the hierarchy in \cite[Thm.~3.1]{berta2022bilinearSDPhierarchy} (with the PPT condition $P^{T_1}\geq 0$ included)
to the optimization problem in \eqref{eq:kpos_bilin_alt} yields the quantity
\begin{multline}\label{eq:kpos_SDP}
    F(\Phi):=\min_{\substack{P\in\mathbb{P}(\mathbb{C}^{dk}\otimes\mathbb{C}^{dk})\\\rho\in\mathbb{D}(\mathbb{C}^d\otimes\mathbb{C}^k)\\\omega\in\mathbb{D}(\mathbb{C}^d\otimes\mathbb{C}^k)}}\Big\{\tr(\mathsf{C}(\Phi\otimes\text{id}_k)P):\tr_1(P)=\omega,\\[-2.5em]\tr_2(P)=\rho,\,P^{T_1}\geq 0\Big\},
\end{multline}
where $\mathbb{P}$ denotes the set of positive semi-definite linear operators, and $(\cdot)^{T_1}=(T\otimes{\rm id})(\cdot)$ is the partial transpose over the first subsystem. In fact, by making the key observation that
\begin{equation}
    F(\Phi)\leq \min_{\rho,\omega\in\mathbb D(\mathbb C^d\otimes\mathbb C^k)}{\rm tr}(\mathsf C(\Phi\otimes{\rm id}_k)(\rho\otimes\omega))
\end{equation}
(because $F(\Phi)$ involves a minimization with respect to a set that contains the optimization domain of \eqref{eq:kpos_bilin_alt}), this leads to the following \textit{sufficient} condition for a Hermitian-preserving linear map $\Phi$ to be $k$-positive:
If $F(\Phi)\geq 0$, then $\Phi$ is $k$-positive, because then
\begin{equation}\label{eq:kpos_SDP_condition}
\min_{\rho,\omega\in\mathbb D(\mathbb C^d\otimes\mathbb C^k)}{\rm tr}(\mathsf C(\Phi\otimes{\rm id}_k)(\rho\otimes\omega)) \geq F(\Phi)\geq 0\,,
\end{equation}
i.e.,~\eqref{eq:kpos_bilin_alt} holds.
To be transparent, when doing a numerical study of this sufficient condition (i.e., generating random Hermitian matrices and checking the condition in \eqref{eq:kpos_SDP_condition}), we only found \textit{decomposable} positive maps that this sufficient criterion could detect as positive, among several thousand trials. Whether this has to do with our underlying sampling method, or whether including the PPT condition $P^{T_1}\geq 0$ somehow tailors this condition more towards detecting decomposable maps, is unclear to us. The code used to implement the optimization in \eqref{eq:kpos_SDP} can be found in the ancillary files of the arXiv submission of this paper.

\section{A new way to sample non-decomposable $k$-positive maps}\label{sec_newnondec}
With this, let us return to the other question raised in the introduction: How can one numerically sample $k$-positive\footnote{For simplicity we will drop the ``$k$-'' and focus on positive maps, but we stress that the method presented in this section works exactly the same for $k$-positivity.} maps?
There are (at least) two straightforward ideas:

1.~Sample a Hermitian matrix at random and check whether it is the Choi matrix of a positive map (however, the constraints here go much beyond the number of negative eigenvalues from Rem.~\ref{rem1}~(ii)).

2.~Fix a map $\Phi$ from the interior of the completely positive maps (w.r.t.~the Hermitian-preserving maps, so $\mathsf C(\Phi)>0$) and, again, sample a Hermitian matrix $H$ at random. Then, take an increasing sequence of step sizes $(t_n)_{n\in\mathbb N}\subset[0,1]$ and check whether $(1-t_n)\mathsf C(\Phi)+t_nH$ is the Choi matrix of a positive map for each $n$. For small enough $n$, this map will always be completely positive, but as $n$ increases one will either land inside the positive maps, or one will pierce a common face of the cone of positive and the cone of completely positive maps \cite{Marciniak10}; so not only is there no guarantee that this approach works, but both these ideas suffer from the fundamental problem that they rely on a way to decide ($k$-)positivity.

\subsection{The theory}
    
What if instead we have a method where all we need is just \textit{one} positive map, and from that map we can produce an infinite number of maps that are, by design, \textit{guaranteed} to be positive (but not completely positive)?
The following proposition---which is motivated by Lie-semigroup theory---will be the foundation for this new strategy.
    \begin{proposition}\label{prop_exp_semigroup}
    Let $\mathsf{CP}(d)\subset\mathcal L(\mathbb C^{d\times d})$ denote the set of all completely positive maps in $d$-dimensions, and let $\mathcal S\subseteq\mathcal L(\mathbb C^{d\times d})$, $\mathcal S\supseteq\mathsf{CP}(d)$ be a closed convex cone which is also a semigroup with identity. Then for all $\Phi\in\mathcal S$, all $K\in\mathbb C^{d\times d}$, and all $t\geq 0$ one has $e^{t(K(\cdot)+(\cdot)K^\dagger+\Phi)}\in\mathcal S$.
    \end{proposition}
    \begin{proof}
    Because $\Phi\in\mathcal S$, so is $\sum_{j=1}^m\frac{t^j}{j!}\Phi^j$ for all $m\in\mathbb N$ (semigroup property \& invariance under positive linear combinations as $t\geq 0$), and hence the same is true for its limit $e^{t\Phi}$ (closedness).
    Next, using the Lie-Trotter formula
    \begin{align*}
    e^{t(K(\cdot)+(\cdot)K^\dagger+\Phi)}&=\lim_{m\to\infty}\Big( e^{\frac tmK(\cdot)+(\cdot)\frac tmK^\dagger}e^{\frac tm\Phi} \Big)^m\\
    &=\lim_{m\to\infty}\Big( e^{\frac tmK}e^{\frac tm\Phi}(\cdot)(e^{\frac tmK})^\dagger \Big)^m
    \end{align*}
    we see that $(e^{\frac tmK}e^{\frac tm\Phi}(\cdot)(e^{\frac tmK}))^m$ is in $\mathcal S$ for all $m\in\mathbb N$, $t\geq 0$ (semigroup property \& $\mathsf{CP}(d)\subseteq\mathcal S$), hence the same is true for its limit (closedness).
    \end{proof}
Now, the basic idea of our method is as follows: start from some ($k$-)positive map $\Phi$ and randomly generate $K\in\mathbb C^{d\times d}$. Then, by Prop.~\ref{prop_exp_semigroup}, the semigroup $e^{t(K(\cdot)+(\cdot)K^\dagger+\Phi)}$ will again be ($k$-)positive for all $t\geq 0$. \textit{However}, it may happen that the resulting one-parameter semigroup is even completely positive for all times, although $\Phi$ itself was not completely positive. Examples for which this happens are the original Choi map \cite{Choi75b}, the Breuer--Hall map \cite{Breuer06,Breuer06b,Hall06}, and the Tomiyama maps \cite{TT83,Tomiyama83,Tomiyama85}.
    As this renders the entire approach pointless we have to identify an additional property on the initial map which guarantees that the corresponding semigroup is not always completely positive. This leads to the next proposition.
    \begin{proposition}\label{prop_CCP}
    Given any $\Phi\in\mathcal L(\mathbb C^{d\times d})$ Hermitian-preserving the following statements are equivalent.
    \begin{itemize}
        \item[(i)] $e^{t(K(\cdot)+(\cdot)K^\dagger+\Phi)}\in\mathsf{CP}(d)$ for all $K\in\mathbb C^{d\times d}$ and all $t\geq 0$
        \item[(ii)] $e^{t(K(\cdot)+(\cdot)K^\dagger+\Phi)}\in\mathsf{CP}(d)$ for some $K\in\mathbb C^{d\times d}$ and all $t\geq 0$
        \item[(iii)] There exists $K'\in\mathbb C^{d\times d}$ such that $\Phi-K'(\cdot)-(\cdot)(K')^\dagger\in\mathsf{CP}(d)$
        \item[(iv)] $\Phi$ is \textnormal{conditionally completely positive}, i.e., $$({\bf1}-|\Omega\rangle\langle\Omega|)\mathsf C(\Phi)({\bf1}-|\Omega\rangle\langle\Omega|)\geq 0$$ with $|\Omega\rangle:=\frac1{\sqrt d}\sum_{j=1}^d|jj\rangle$ the maximally entangled state.
    \end{itemize}
    \end{proposition}
    \begin{proof}
    ``(i) $\Rightarrow$ (ii)'': Trivial. ``(ii) $\Rightarrow$ (iii)'': Well-known form of generators of completely positive one-parameter semigroups \cite[Thm.~3.1]{CE79}. ``(iii) $\Rightarrow$ (i)'': Prop.~\ref{prop_exp_semigroup} for $\mathcal S=\mathsf{CP}(d)$.
    ``(iii) $\Leftrightarrow$ (iv)'': \cite{Wolf08b} or \cite[Thm.~14.7]{EL77}.
    \end{proof}
In Figure~\ref{figa} we give some geometric intuition regarding how positivity, complete positivity, and conditional complete positivity are related.

    \begin{figure}[!htb]
\centering
\includegraphics[width=0.5\textwidth]{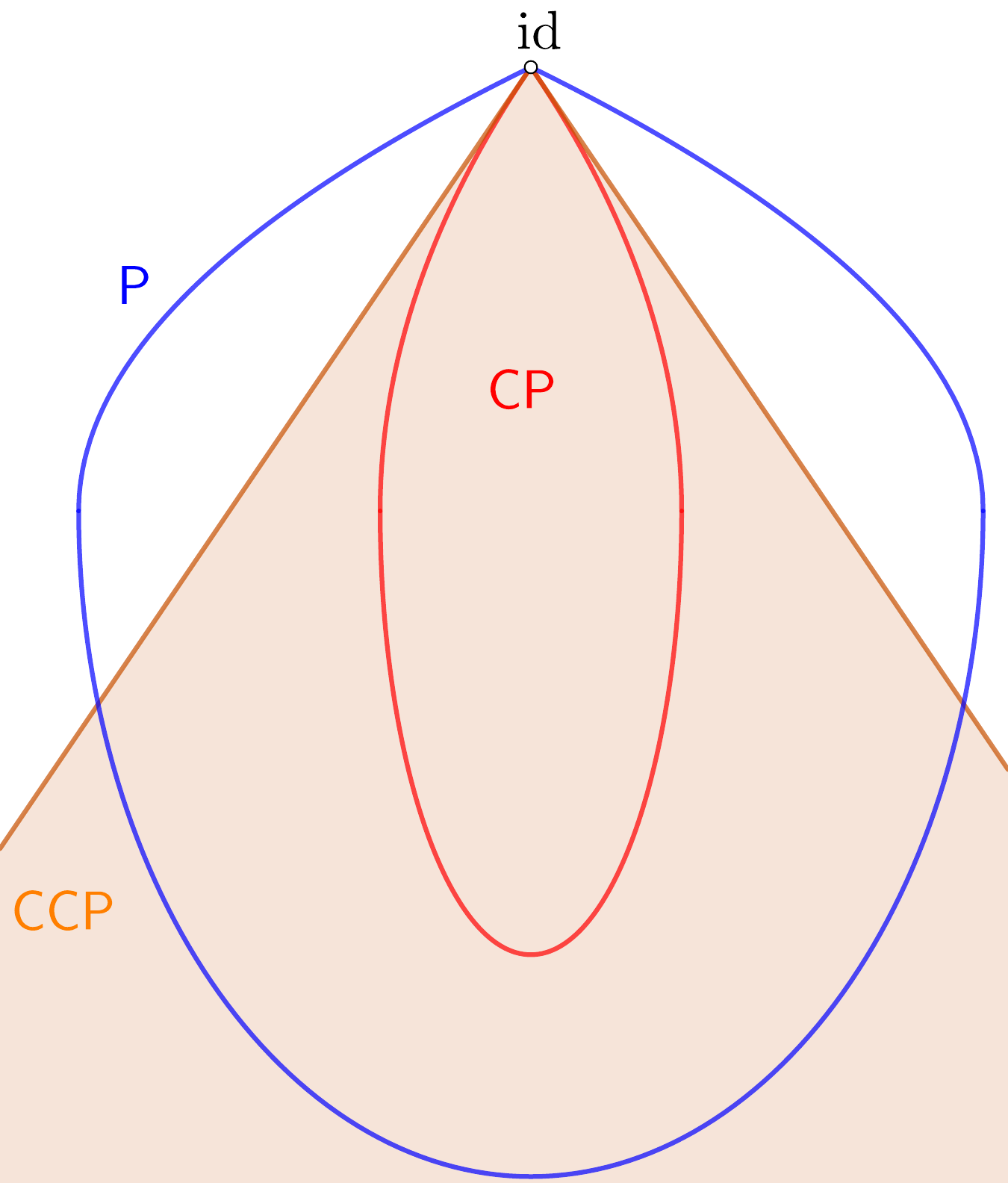}
\caption{
The set of conditionally completely positive maps (CCP) is the tangent cone of the set of completely positive (CP) maps at the identity ${\rm id}$ \cite{DHKS08}. Moreover, the CP cone sits strictly inside the cone of positive (P) maps. Therefore, a non-trivial portion of P is in CCP$\setminus$CP, meaning that if any element of P$\,\cap\,$CCP is taken as the generator in Prop.~\ref{prop_exp_semigroup}, then (and only then) will the corresponding semigroup be trivial (i.e., CP) for all positive times.
}\label{figa}
\end{figure}

\subsection{The workflow}

To summarize, we have a simple method to generate maps which are, by design, positive but not completely positive\footnote{From here on out, when we use the term ($k$-)positive we implicitly also mean that the map should not be completely positive.}. All we need for it is a starting point, that is, one positive map which is not conditionally completely positive. But for now it is not even clear that such a map exists (recall the examples before Prop.~\ref{prop_CCP}).
The next most obvious pick would be the transposition map: For $K=\frac{\bf1}2$, this results in $e^{t( T-{\rm id})}=\frac{1+e^{-2t}}2{\rm id}+\frac{1-e^{-2t}}2 T$ \cite[Ex.~6]{Chru14}. One might think that the trivial form of this semigroup is due to the poor choice of $K$, and to some extent that is correct, but this example hints at a larger problem we have not considered so far: We wanted our method to generate \textit{non-decomposable} positive maps; however, the set of decomposable maps satisfies all assumptions of Prop.~\ref{prop_exp_semigroup} (i.e., is a closed convex cone which is also a semigroup with identity, and contains all completely positive maps), as is readily verified. Hence, if the starting point $\Phi$ is decomposable, then the corresponding semigroup will always be decomposable, regardless of the chosen matrix $K$. This insight rules out not only the transposition map as a starting point, but also the reduction map $X\mapsto \tr(X){\bf1}-X$ (its transpose is completely positive) \cite{HH99}.
Putting everything together, this lets us arrive at the definitive workflow of our new method.
\begin{workflow}\label{w_newmethod}
\begin{enumerate}[(a)]
    \item Start with a map $\Phi$ which is $k$-positive, non-decomposable\footnote{
    Strictly speaking, we need $\Phi$ to not be ``conditionally decomposable'', i.e., there should not exist any $K\in\mathbb C^{d\times d}$ such that $\Phi-K(\cdot)-(\cdot)K^\dagger$ is decomposable \cite[Thm.~3]{Szczygielski23}.
    However, unlike with conditional complete positivity we are not aware of a simple test for this, which is why we instead settle for the \textit{necessary} condition of $\Phi$ being non-decomposable.
    \label{footnote_cond_decomp}}, and not conditionally completely positive
    \item Generate a random matrix $K$
    \item Compute $e^{t(K(\cdot)+(\cdot)K^\dagger+\Phi)}$. By (a), there exists $t_0>0$ such that this semigroup is $k$-, but not completely positive for all $t\in[0,t_0)$ (and finding $t_0$ is simple via the Choi matrix)
    \item Check that the semigroup is non-decomposable for small enough $t$ (this is an SDP, hence feasible to check, cf.~Sec.~\ref{sec_ppt2}).
    \item If the previous two points were successful, then we have found $t_0>0$ for which $e^{t_0(K(\cdot)+(\cdot)K^\dagger+\Phi)}$ is a non-decomposable $k$-positive map. If we also check conditional complete positivity (and perhaps lower $t_0$ accordingly), this map can be used as a new starting point for step (a).
\end{enumerate}
\end{workflow}
The last point shows that this new method even has the potential to \textit{indefinitely} produce new classes of non-decomposable positive maps.
This is in contrast to the only other method we are aware of to generate non-decomposable positive maps\footnote{Here and henceforth, when we say non-decomposable we implicitly mean that the map should be positive, as well.}, which is to start from a $k$-positive map $\Phi$ and tensor it with the $k$-dimensional identity (i.e., ${\rm id}_k\otimes\Phi$ is positive and non-decomposable) \cite{HLLM18}; notably, this method only works once and it also changes the dimension, perhaps drastically, depending on $k$.

All that is left now is to prove that this method works, i.e., we have to find a map which is neither decomposable nor conditionally completely positive, and then show that the corresponding semigroup is non-decomposable for small enough times. One of the very few papers we found which contains a map (or rather, a one-parameter family of maps) with these properties is \cite{Bhattacharya21}. Although going through Workflow~\ref{w_newmethod} with \cite[Eq.~(1)]{Bhattacharya21} 
seems to always give something decomposable (recall footnote~\ref{footnote_cond_decomp}), we were able to slightly tweak the map in question such that our method works as intended. The map
\begin{equation}\label{eq:bhatta_alt_dim3}
\Lambda_{\alpha}'(X) = \frac{1}{\alpha + \frac{1}{\alpha}}\begin{pmatrix}
    \alpha (x_{11}+x_{22}) & x_{12} & \alpha x_{13} \\
x_{21} & \frac{x_{22} + x_{33}}{\alpha} & x_{32} \\
\alpha x_{31} & x_{23} & \alpha x_{33}+\frac{x_{11}}{\alpha}
\end{pmatrix}
\end{equation}
is non-decomposable and not conditionally completely positive for all $\alpha\in(0,1]$. While the proof is analogous to the proof in \cite{Bhattacharya21}, we nonetheless present the key arguments in Appendix~D for the sake of completeness.

Here is the summary of the numerics we conducted with this map for the special case $\alpha=1$:
First, when setting $K=0$, the semigroup $(e^{t\Lambda_1'})_{t\geq 0}$ crosses into the completely positive cone at $t\approx 1.95$ and, more importantly, is non-decomposable for all $0<t\lessapprox 0.44$. When repeating this for randomly chosen $K$ the latter threshold fluctuated around $\approx 0.25-0.35$.
In a second step, we took $\tilde\Lambda:=e^{0.2\Lambda_1'}$ as the starting point for another iteration of Workflow~\ref{w_newmethod} and we found that, again for $K=0$ for simplicity, the semigroup $(e^{t\tilde\Lambda})_{t\geq 0}$ crosses into the completely positive cone at $t\approx 7.7$ and is non-decomposable for all $0<t\lessapprox 1.1$ (resp.~$\lessapprox 0.25-0.4$ for random $K$'s).
We tried this for a few more layers and found that this behavior continues, thus showing that this method can, with success, be used recursively. 

Another  advantage of~\eqref{eq:bhatta_alt_dim3} is 
a rather simple idea behind it. Namely,  transpose just a sub-block of the matrix and mix the diagonal, such that the resulting map is positive.
This allows for a generalization to higher dimensions (for simplicity for $\alpha=1$ and $d=4$): $\Phi:\mathbb C^{4\times 4}\to\mathbb C^{4\times 4}$, with $\Phi(X)$ defined via
\begin{equation}\label{eq:bhatta_gen_d4}
    \frac13{\small\begin{pmatrix}
        x_{11}+x_{22}+x_{33}&x_{12}&x_{13}&x_{14}\\
        x_{21}&x_{22}+x_{33}+x_{44}&x_{23}&x_{24}\\
        x_{31}&x_{32}&x_{33}+x_{44}+x_{11}&x_{43}\\
        x_{41}&x_{42}&x_{34}&x_{44}+x_{11}+x_{22}
    \end{pmatrix}}
\end{equation}
is almost certainly\footnote{Unlike with~\eqref{eq:bhatta_alt_dim3} we did not prove that~\eqref{eq:bhatta_gen_d4} is positive. However, we  numerically sampled five million vectors $\phi\in\mathbb C^4$ and verified that $\Phi(|\phi\rangle\langle\phi|)\geq 0$ every time. From this we conclude that with sufficiently high probability $\Phi$ is in fact positive.}
positive.
As in the previous case, we find that $e^{t\Phi}$ is non-decomposable until $t\approx 0.24$.\footnote{The code used to generate the maps is available in the ancillary files of the arXiv submission of the paper.}

% https://colab.research.google.com/drive/1JjfRdsb893hPzsF9gZZDaw_D6KnJygeu?authuser=1#scrollTo=fRaMcEg6NQh7

Finally, the only other paper we are aware of which features a positive map that is neither decomposable nor conditionally completely positive is \cite{MO15}, resp.~\cite{RSC15} where it is generalized to arbitrary finite dimensions. From our testing, it also constitutes a viable input to Workflow~\ref{w_newmethod}~(a) in that the corresponding semigroup is positive and non-decomposable for small enough times.

\section{The PPT-square conjecture}\label{sec_ppt2}

As mentioned in the introduction, our method for generating positive maps can be used to test the PPT-square conjecture~\cite{RuskaiEtAl2012}, particularly in higher dimensions where the conjecture remains open.
To do so, we make use of the following formulation of the conjecture~\cite[Conjecture~IV.3]{CMW19}: For any completely positive and completely co-positive map $\Delta$ and any positive map $\Psi$, the composition $\Psi\circ\Delta$ is decomposable.

Decomposability, i.e., $\Phi=\Psi_1+\Psi_2\circ T$ for some $\Psi_1$, $\Psi_2$ completely positive, can be straightforwardly checked, via semi-definite programming, as a feasibility problem. One can make it into an optimization problem via a slack variable, as follows. Let $\Phi$ be a Hermitian-preserving linear map. Define the quantity
\begin{equation*}
    D(\Phi):=\min\Big\{\|S\|_1:\mathsf{C}(\Phi)=P_1+P_2^{T_2}+S,\,P_1\geq 0,\,P_2\geq 0\Big\},
\end{equation*}
where $\|\cdot\|_1$ is the trace norm, and $(\cdot)^{T_2}=({\rm id}\otimes T)(\cdot)$ is the partial transpose on the second subsystem. This is indeed a semi-definite program, on account of the fact that the trace norm of an arbitrary Hermitian operator $S$ can be expressed as $\|S\|_1=\min\big\{\tr[Y]:-Y\leq S\leq Y\big\}$, which is itself a semi-definite program. It is clear that $D(\Phi)\geq 0$ for all $\Phi$. In particular, we can use $D(\Phi)$ to determine decomposability, because
\begin{equation*}
    \Phi\text{ decomposable} \Longleftrightarrow D(\Phi)=0
\end{equation*}
(equivalently, $\Phi$ is non-decomposable if and only if $ D(\Phi)>0$).
Now, the aforementioned formulation of the PPT-square conjecture implies the following for a given positive map $\Psi$:
\begin{align*}
    \text{PPT-square conjecture true}&\,\Longrightarrow D(\Psi\circ\Delta)=0\quad\forall~\Delta\in\mathsf{CP}\cap\mathsf{CcoP}\\
    &\Longleftrightarrow \min_{\Delta\in\mathsf{CP}\cap\mathsf{CcoP}} D(\Psi\circ\Delta)=0,
\end{align*}
where $\mathsf{CcoP}$ denotes the set of completely co-positive linear maps. In other words, a given positive (non-decomposable) map $\Psi$ induces a counterexample to PPT-square if and only if $\min_{\Delta\in\mathsf{CP}\cap\mathsf{CcoP}}D(\Psi\circ\Delta)>0$. Note that this optimization is also an SDP:
\begin{multline}
    \min_{\Delta\in\mathsf{CP}\cap\mathsf{CcoP}}D(\Psi\circ\Delta)=\min\Big\{\|S\|_1:\mathsf{C}(\Psi)\star N=P_1+P_2^{T_2}+S,\,P_1\geq 0,\\P_2\geq 0,\,N\geq 0,\,N^{T_2}\geq 0\Big\},\label{eq:PPT2_SDP}
\end{multline}
where the variable $N$ in the optimization is the Choi representation of a completely positive and completely co-positive map $\Delta$, and
\begin{equation}
    \mathsf{C}(\Psi)\star N\equiv\tr_2\big[\big(N^{T_2}\otimes\mathbbm{1}\big)\big(\mathbbm{1}\otimes \mathsf C(\Psi)\big)\big]
\end{equation}
is the Choi representation of the composition $\Psi\circ\Delta$. 

By performing the optimization in \eqref{eq:PPT2_SDP} using the positive, non-decomposable maps $\Psi$ generated in Section~\ref{sec_newnondec}, after going through several hundred examples we were unfortunately unable to find a counterexample. Our code can be found in the ancillary files of the arXiv submission of our paper.

\section{Outlook}\label{sec_outlook}

The aims of this work were twofold. 
First, to analyze and extend the new characterization of $k$-positivity~\cite{MOM25} formulated via the completely positive map associated to a given positive map. We also  wanted to leverage this formulation to uncover structural and computational links -- specifically, to norms of order-3 tensors and to the separability problem. In particular, we encoded $k$-positivity as a family of bilinear optimization problems, which provides a basis for numerical procedures and yields new sufficient criteria for a Hermitian-preserving map to be $k$-positive.  One of the more obvious open questions here is: What is the best way to decide $k$-positivity in finite time (e.g., in a complexity sense)? After all, just because the corresponding decision problem is most likely NP-hard, this does not mean that some of the criteria, respectively the corresponding algorithms to compute the associated quantities, cannot scale better than others (e.g., exponential versus doubly exponential).

The second aim was to design a new Lie-semigroup inspired method to generate maps that are guaranteed to be non-decomposable ($k$-positive), so long as the map that is used to ``feed'' the corresponding algorithm has suitable properties (Workflow~\ref{w_newmethod}).
We illustrated this method with an example for $d=3$ and $d=4$ which numerically produced new non-decomposable maps, even recursively. Using a new SDP-based test of the PPT-square conjecture we, however, did not find any violation of the conjecture. It should be noted though  that there is an inherent flaw to this method, which is that it can never produce any non-Markovian maps because it relies on one-parameter semigroups. 
Thus, an interesting follow-up question appears: What portion of the Markovian non-decomposable maps can be explored with the new method? After all, the further from the identity a target map is, the less likely it is that there is a semigroup that connects the map to the identity -- and if the non-decomposable maps have any disconnected components, then these can, by design, not be explored with our new method. One more thing to mention is that if one wants to find non-decomposable \textit{trace-preserving} maps, all one has to do is restrict $K$ in Workflow~\ref{w_newmethod} to $K=iH-\frac12\Phi^*({\bf1})$ and instead sample random Hermitian matrices $H$.

Finally, although Workflow~\ref{w_newmethod} ensures that the generated semigroup remains in the same $k$-positivity class as the seed, it does not specify when---or even whether---the trajectory passes through intermediate classes on the way to complete positivity. The obstacle is the lack of an \emph{effective} certification procedure (even in modest dimensions) for $(k{+}1)$-, $(k{+}2)$-, $\ldots$, $(d{-}1)$-positivity.\footnote{By contrast, \emph{refuting} $k$-positivity is much easier: It suffices to find a vector $\psi$ of Schmidt rank $k$ with $\langle\psi\,|\,\mathsf C(\Phi)\,|\,\psi\rangle<0$; see Prop.~\ref{prop_CP_P_char}.} This limitation is especially relevant when searching for violations of the Sanpera--Bru\ss--Lewenstein conjecture because in order to generate a non-decomposable $(d{-}1)$-positive map, unless we have an effective method to certify $d{-}1$ positivity we have to start from a seed that is $(d{-}1)$-positive and non-decomposable, thus leading to a circular dependence.

\section*{Acknowledgments}

We thank Alexander M\"uller-Hermes for clarifying remarks on the PPT-square and the Sanpera conjecture \cite{MH25_priv}, as well as Dariusz Chru{\'s}ci{\'n}ski for making us aware of the papers \cite{MO15,RSC15} as possible inputs for our Workflow~\ref{w_newmethod}.
FvE is funded by the \textit{Deutsche Forschungsgemeinschaft} (DFG, German Research Foundation) -- project number 384846402, and supported by the Einstein Foundation (Einstein Research Unit on Quantum Devices) and the MATH+ Cluster of Excellence. SD acknowledges support from COST Action CA24109 – ``Many-body Open Quantum Systems (QOpen)''.

%F: DFG-Forschungsprojekt "TP P4" der Forschungsgruppe 2724, https://gepris.dfg.de/gepris/projekt/408788160

\section*{Appendix A: Proof of Proposition~\ref{prop_CP_P_char}}

\begin{proof}
``(i) $\Leftrightarrow$ (ii)'':
Follows from \cite[Prop.~3]{CK09}.
``(ii) $\Leftrightarrow$ (iii)'': A straightforward computation shows $\langle z|\mathsf C(\Phi)|z\rangle=
\langle\overline{{\rm vec}^{-1}(z)}\otimes{\rm vec}^{-1}(z)|\widehat\Phi\rangle_{\rm HS}$ for all $z$; one way to see this is to use that every linear map $\Phi$ can be written as $\sum_i A_i(\cdot)B_i^\dagger$ \cite[Coro.~2.21]{Watrous18}, together with the facts $\mathsf C(A(\cdot)B^\dagger)=|{\rm vec}(A)\rangle\langle{\rm vec}(B)|$ and $\widehat{A(\cdot)B^\dagger}=\overline B\otimes A$.
Moreover, $z\in\mathbb C^m\otimes\mathbb C^d$ has Schmidt rank $k$ if and only if ${\rm vec}^{-1}(z)$ has rank $k$ as the Schmidt decomposition is just a vectorized singular value decomposition.
``(iii) $\Leftrightarrow$ (iv)'' Combine the previous equivalence with \cite[Lemma~2]{vE24_decomp_findim}.
``(iv) $\Leftrightarrow$ (v)'': Because $\tr(\Phi(X(\cdot)X^\dagger))\geq 0$ for all $X\in\mathbb C^{m\times d}$ with rank at most $k$, and because $X(\cdot)X^\dagger$ is completely positive, this condition readily extends---by linearity---to all completely positive maps which admit a set of Kraus operators\footnote{
 The convention we choose here and henceforth is that $\{K_j\}_j$ are Kraus operators of $\Phi$ if $\Phi=\sum_jK_j(\cdot)K_j^\dagger$.
 }
 of rank at most $k$.
\end{proof}

\section*{Appendix B: Proof of Lemma~\ref{lemma_char_KyFan}}

\begin{proof}
In what follows let $B=\sum_jb_j|g_j\rangle\langle g_j|$ be any diagonalization of $B$ where $b_1\geq b_2\geq \ldots\geq 0$ and $\{g_j\}_j\subset\mathbb C^d$ is an orthonormal basis. This is possible because $B\geq0$ by assumption.
``$\leq$'': Defining $M_B:=\sum_{j=1}^k|g_j\rangle\langle g_j|$ one finds
\begin{align*}
    \|B\|_{(k)}=\sum_{j=1}^kb_j=\tr(BM_B)\leq \max_{\substack{M\in\mathbb C^{d\times d}\\0\leq M\leq{\bf1},\tr(M)\leq k}}
 \tr(BM)\,,
\end{align*}
as desired.
``$\geq$'': Let any $M\in\mathbb C^{d\times d}$ be given such that $0\leq M\leq{\bf1}$ and $\tr(M)\leq k$. Moreover, let $m_1\geq m_2\geq\ldots\geq 0$ denote the eigenvalues of $M$. Then
$
\tr(BM)\leq \sum_j b_jm_j
$
as a consequence of von Neumann's trace inequality \cite[Ch.~9, Thm.~H.1.g]{MarshallOlkin}.
This implies
\begin{align*}
\tr(BM)&\leq\sum_{j=1}^{k-1}b_jm_j+\sum_{j=k}^db_jm_j\\
&\leq \sum_{j=1}^{k-1}b_jm_j+b_k\sum_{j=k}^dm_j=\sum_{j=1}^{k-1}b_jm_j+b_k\Big(\tr(M)-\sum_{j=1}^{k-1}m_j\Big)\,.
\end{align*}
Next we use that $\tr(M)\leq k\,$:
\begin{align*}
 \tr(BM)&\leq \sum_{j=1}^{k-1}b_jm_j+b_k\Big(k-\sum_{j=1}^{k-1}m_j\Big)=kb_k+\sum_{j=1}^{k-1}(b_j-b_k)m_j
 \end{align*}
 Finally---using that $m_j\leq 1$ due to $M\leq{\bf1}$ because assumption---as $b_j-b_k\geq 0$ for all $j\leq k$ we can estimate
 \begin{align*}
\tr(BM)&\leq kb_k+\sum_{j=1}^{k-1}(b_j-b_k)=kb_k+\sum_{j=1}^{k-1}b_j-(k-1)b_k=\sum_{j=1}^{k}b_j=\|B\|_{(k)}\,.
\end{align*}
Because $M$ was chosen arbitrarily this yields the desired inequality.
\end{proof}
\section*{Appendix C: Proof of Theorem~\ref{thm1a}}
\begin{proof}
Obviously, a Hermitian-preserving map $\Phi$ is $k$-positive if and only if
$$
\min_{\substack{x,y\in\mathbb C^d\otimes\mathbb C^k\\ \|x\|=\|y\|=1}}\langle x|(\Phi\otimes{\rm id}_k)(|y\rangle\langle y|)|x\rangle\geq 0\,.
$$
Equivalently, $ k\max( \sigma( \mathsf C(\Phi) ) )-\min_{x,y\in\mathbb C^d\otimes\mathbb C^k,\|x\|=\|y\|=1}\langle x|(\Phi\otimes{\rm id}_k)(|y\rangle\langle y|)|x\rangle$ cannot exceed $k\max( \sigma( \mathsf C(\Phi) ) )$. The expression on the left-hand side can be re-written further:
\begin{align}
    k&\max( \sigma( \mathsf C(\Phi) ) )-\min_{\substack{x,y\in\mathbb C^d\otimes\mathbb C^k\\ \|x\|=\|y\|=1}}\langle x|(\Phi\otimes{\rm id}_k)(|y\rangle\langle y|)|x\rangle\notag\\
    &=\max_{\substack{x,y\in\mathbb C^d\otimes\mathbb C^k\\ \|x\|=\|y\|=1}}\big(k\max( \sigma( \mathsf C(\Phi) ) )-\langle x|(\Phi\otimes{\rm id}_k)(|y\rangle\langle y|)|x\rangle\big)\notag\\
    &=\max_{\substack{x,y\in\mathbb C^d\otimes\mathbb C^k\\ \|x\|=\|y\|=1}}\big(k\max( \sigma( \mathsf C(\Phi) ) )\tr(|y\rangle\langle y|)\langle x|{\bf1}|x\rangle-\langle x|(\Phi\otimes{\rm id}_k)(|y\rangle\langle y|)|x\rangle\big)\notag\\
    &=\max_{\substack{x,y\in\mathbb C^d\otimes\mathbb C^k\\ \|x\|=\|y\|=1}}\Big\langle x\Big|\Big(k\max( \sigma( \mathsf C(\Phi) ) )\tr(\cdot){\bf1}-\Phi\otimes{\rm id}_k\Big)(|y\rangle\langle y|)\Big|x\Big\rangle\label{eq:thm1a_2}
\end{align}
Next, we claim that $k\max( \sigma( \mathsf C(\Phi) ) )\tr(\cdot){\bf1}-\Phi\otimes{\rm id}_k$ is completely positive: First,
\begin{align}
    \mathsf C\Big(k\max( \sigma( \mathsf C(\Phi) ) )\tr(\cdot){\bf1}-\Phi\otimes{\rm id}_k\Big)
&=k\max( \sigma( \mathsf C(\Phi) ) ){\bf 1}-\mathsf C(\Phi\otimes{\rm id}_k)\,.\label{eq:thm2_1}
\end{align}
Moreover, $\mathsf C(\Phi\otimes{\rm id}_k)=({\bf1}\otimes\mathbb F\otimes{\bf1})(\mathsf C(\Phi)\otimes\mathsf C({\rm id}_k))({\bf1}\otimes\mathbb F\otimes{\bf1})$ (with $\mathbb F$ the flip on the second and third system) as is readily verified, meaning the largest eigenvalue of $\mathsf C(\Phi\otimes{\rm id}_k)$ equals the largest eigenvalue of $\mathsf C(\Phi)\otimes\mathsf C({\rm id}_k)=\mathsf C(\Phi)\otimes|\Gamma\rangle\langle\Gamma|$ with $|\Gamma\rangle=\sum_{j=1}^k|jj\rangle$ the (unnormalized) maximally entangled state in $k$ dimensions. But this eigenvalue, in turn, is just $\max(\sigma(\mathsf C(\Phi)))\|\Gamma\|^2=k\max(\sigma(\mathsf C(\Phi)))$ which shows that~\eqref{eq:thm2_1} is positive semi-definite (because $\mathsf C(\Phi)$ is Hermitian by assumption), as desired.
Thus, complete positivity of said map implies the existence of $V:\mathbb C^d\otimes\mathbb C^k\to\mathbb C^d\otimes\mathbb C^k\otimes\mathbb C^r$ linear such that $k\max( \sigma( \mathsf C(\Phi) ) )\tr(\cdot){\bf1}-\Phi\otimes{\rm id}_k=\tr_3(V(\cdot)V^\dagger)$ \cite[Thm.~2.22]{Watrous18}. Hence~\eqref{eq:thm1a_2} is equal to
\begin{align*}
    \max_{\substack{x,y\in\mathbb C^d\otimes\mathbb C^k\\ \|x\|=\|y\|=1}}\langle x|\tr_3(V|y\rangle\langle y|V^\dagger)|x\rangle&=
    \max_{\substack{x,y\in\mathbb C^d\otimes\mathbb C^k\\ \|x\|=\|y\|=1}}\tr\big(|x\rangle\langle x|\tr_3(V|y\rangle\langle y|V^\dagger)\big)\\
    &=\max_{\substack{x,y\in\mathbb C^d\otimes\mathbb C^k\\ \|x\|=\|y\|=1}}\tr\big((|x\rangle\langle x|\otimes{\bf1})V|y\rangle\langle y|V^\dagger\big)\\
    &=\max_{\substack{x,y\in\mathbb C^d\otimes\mathbb C^k\\ \|x\|=\|y\|=1}}\langle y|V^\dagger(|x\rangle\langle x|\otimes{\bf1})V|y\rangle\\
    &=\max_{x\in\mathbb C^d\otimes\mathbb C^k,\|x\|=1}\|V^\dagger(|x\rangle\langle x|\otimes{\bf1})V\|_\infty
\end{align*}
In the last step we used that $V^\dagger(|x\rangle\langle x|\otimes{\bf1})V\geq 0$. But $$V^\dagger(|x\rangle\langle x|\otimes{\bf1})V=V^\dagger(|x\rangle\otimes{\bf1})(\langle x|\otimes{\bf1})V=V^\dagger(|x\rangle\otimes{\bf1})(V^\dagger(|x\rangle\otimes{\bf1}))^\dagger$$
so---using the previously employed $C^*$-identity $\|AA^\dagger\|_\infty=\|A\|_\infty^2$---this shows that~\eqref{eq:thm1a_2} is the maximum of $\|V^\dagger(|x\rangle\otimes{\bf1})\|_\infty^2$ taken over all unit vectors $x$.
In summary, $\Phi$ is $k$-positive if and only if
$$
\max_{x\in\mathbb C^d\otimes\mathbb C^k,\|x\|=1}\|V^\dagger(|x\rangle\otimes{\bf1})\|_\infty^2\leq  k\max( \sigma( \mathsf C(\Phi) ) )
$$
which, in turn, is equivalent to condition~(ii). This concludes the proof.
\end{proof}

\begin{remark}
    It is worth pointing out that, as this proof shows,
    the difference between $ k\max( \sigma( \mathsf C(\Phi) ) )$ and $\max_{x\in\mathbb C^d\otimes\mathbb C^k,\|x\|=1}\|V^\dagger(| x\rangle\otimes{\bf1})\|_\infty^2$ is given exactly by the smallest value $\langle x|(\Phi\otimes{\rm id}_k)(|y\rangle\langle y|)|x$ can take (when varying $x,y$ over the unit sphere).
    In particular, $$\max_{x\in\mathbb C^d\otimes\mathbb C^k,\|x\|=1}\|V^\dagger(| x\rangle\otimes{\bf1})\|_\infty=\sqrt{ k\max( \sigma( \mathsf C(\Phi) ) )}$$ if and only if there exist non-zero vectors $x,y$ such that $\langle y|\Phi(|x\rangle\langle x|)|y\rangle=0$.
    As an example, for the transposition map $\langle 0|(|1\rangle\langle 1|)^T|0\rangle=0$ so $$\max_{x\in\mathbb C^d\otimes\mathbb C^k,\|x\|=1}\|V^\dagger(| x\rangle\otimes{\bf1})\|_\infty=\sqrt{ k\max( \sigma( \mathsf C(\Phi) ) )}=\sqrt k$$ for all $d\geq 2, 1\leq k\leq d$.
\end{remark}

\section*{Appendix D: Proof of properties of~\eqref{eq:bhatta_alt_dim3}}

First, it suffices to check positivity on the extreme points of the cone of positive semi-definite matrices, i.e., $\Lambda_\alpha'(|\phi\rangle\langle\phi|)\geq 0$ for all $\phi\in\mathbb C^3,\alpha\in(0,1]$. Using Sylvester's criterion \cite[Thm.~7.2.5~(a)]{HJ1_2nd_ed}, this boils down to verifying that all principal minors of the Hermitian matrix
\begin{align*}
\Lambda_\alpha'(|\phi\rangle\langle\phi|)=\frac{1}{\alpha+\frac1\alpha}\begin{pmatrix}
    \alpha (|\phi_1|^2+|\phi_2|^2) & \phi_1\phi_2^* & \alpha \phi_1\phi_3^* \\
\phi_2\phi_1^* & \frac{|\phi_2|^2 + |\phi_3|^2}{\alpha} & \phi_3\phi_2^* \\
\alpha \phi_3\phi_1^* & \phi_2\phi_3^* & \alpha |\phi_3|^2+\frac{|\phi_1|^2}{\alpha}
\end{pmatrix}
\end{align*}
are non-negative. Without loss of generality we may assume $\phi_1,\phi_2,\phi_3\neq 0$ (this special case then extends to the general case by continuity). Now the smallest principle minors are the diagonal entries, which are obviously non-negative.
Next are the principle submatrices of size two:
\begin{align*}
    \begin{vmatrix}
        \alpha (|\phi_1|^2+|\phi_2|^2) & \phi_1\phi_2^*  \\
\phi_2\phi_1^* & \frac{|\phi_2|^2 + |\phi_3|^2}{\alpha} 
    \end{vmatrix}=|\phi_2|^4+|\phi_1|^2|\phi_3|^2+|\phi_2|^2|\phi_3|^2\geq 0\,,
\end{align*}
and similarly for the other two. Finally, the determinant of $\Lambda_\alpha'(|\phi\rangle\langle\phi|)$ comes out to be
{\small\begin{align*}
    \frac1\alpha|\phi_1|^2|\phi_2|^4+\frac1\alpha|\phi_1|^4|\phi_3|^2+\alpha|\phi_2|^2|\phi_3|^4+2\alpha|\phi_1|^2\Re(\phi_2^2(\phi_3^*)^2)+\Big(\frac1\alpha-2\alpha\Big)|\phi_1\phi_2\phi_3|^2.
\end{align*}}
Now we lower bound $\Re(\phi_2^2(\phi_3^*)^2)\geq -|\phi_2|^2|\phi_3|^2$ as well as $\frac1\alpha\geq\alpha$ (because $\alpha\in(0,1]$) to find
\begin{align*}
\det(\Lambda_\alpha'(|\phi\rangle\langle\phi|))&\geq \alpha |\phi_1|^2|\phi_2|^4+\alpha|\phi_1|^4|\phi_3|^2+\alpha|\phi_2|^2|\phi_3|^4-3\alpha|\phi_1\phi_2\phi_3|^2\\
&= 3\alpha|\phi_1\phi_2\phi_3|^2\Big( \frac{ \frac{|\phi_2|^2}{|\phi_3|^2}+\frac{|\phi_1|^2}{|\phi_2|^2}+\frac{|\phi_3|^2}{|\phi_1|^2} }{3}-1 \Big)\\
&\geq 3\alpha|\phi_1\phi_2\phi_3|^2\Big(  \sqrt[3]{\frac{|\phi_2|^2}{|\phi_3|^2}\cdot\frac{|\phi_1|^2}{|\phi_2|^2}\cdot\frac{|\phi_3|^2}{|\phi_1|^2}}-1\Big)=0\,.
\end{align*}
In the last step, we used the inequality between arithmetic and geometric mean (AM-GM). Altogether, this shows that $\Lambda_\alpha'(|\phi\rangle\langle\phi|)\geq 0$ for all $\phi$ and all $\alpha\in(0,1]$, hence $\Lambda_\alpha'$ is positive, as claimed.

Next, for lack of conditional complete positivity we note that $\mathsf C(\Lambda_\alpha')$ has
$$
\begin{pmatrix}
    0&1\\1&\frac1\alpha
\end{pmatrix}
$$
as a (obviously indefinite) principal submatrix. This is not affected by multiplying with $({\bf1}-|\Omega\rangle\langle\Omega|)$, and hence it is also present in $({\bf1}-|\Omega\rangle\langle\Omega|)\mathsf C(\Lambda_\alpha')({\bf1}-|\Omega\rangle\langle\Omega|)$. Thus Prop.~\ref{prop_CCP} shows that $\Lambda_\alpha'$ is not conditionally completely positive.

Finally, let us verify that $\Lambda_\alpha'$ is non-decomposable. This is done using the following family of PPT states:
\begin{align*}
 \tau_x:=\frac3{1+x+x^{-1}}   \begin{pmatrix}
        1&0&0&0&1&0&0&0&1\\
        0&x&0&0&0&0&0&0&0\\
        0&0&x^{-1}&0&0&0&0&0&0\\
        0&0&0&x^{-1}&0&0&0&0&0\\
        1&0&0&0&1&0&0&0&1\\
        0&0&0&0&0&x&0&0&0\\
        0&0&0&0&0&0&x&0&0\\
        0&0&0&0&0&0&0&x^{-1}&0\\
        1&0&0&0&1&0&0&0&1
    \end{pmatrix}\,,
\end{align*}
$x>0$. The principal submatrix of
$({\rm id}\otimes\Lambda_\alpha')(\tau_x)$ consisting of rows 1, 5, and 9 is readily verified to (up to global factors) be given by
\begin{align*}
    \begin{pmatrix}
 \alpha (x+1) & 1 & \alpha \\
 1 & \frac{x+1}{\alpha} & 0 \\
 \alpha & 0 & \frac{x}{\alpha}+\alpha 
    \end{pmatrix}\,,
\end{align*}
so the corresponding determinant (principal minor) is $\frac{(x+2) x^2}{\alpha}+\alpha (x^2+x-1)$. Note that in the limit $x\to 0^+$ this becomes $-\alpha$ (which is $<0$ by assumption) meaning there exists $x_0=x_0(\alpha)>0$ such that $({\rm id}\otimes\Lambda_\alpha')(\tau_{x_0})$ is not a state. But because $\tau_{x_0}$ is a PPT state, this means that $\Lambda_\alpha'$ cannot be decomposable which concludes the proof.

\bibliographystyle{mystyle}
\bibliography{control21vJan20.bib}
\end{document}